\documentclass[smallextended]{svjour3}
\usepackage{amsmath,amssymb}
\usepackage{graphicx}
\usepackage{natbib}
\usepackage{color}
\renewcommand{\cite}{\citep}

\newcommand{\bx}{\mathbf{x}}
\newcommand{\br}{\mathbf{r}}
\newcommand{\sgn}{\operatorname{sgn}}

\title{Evolutionary Games with Affine Fitness Functions: Applications to Cancer}
\author{Moritz Gerstung \and Hani Nakhoul \and Niko Beerenwinkel}
\institute{Department of Biosystems Science and Engineering, ETH Zurich, Mattenstrasse 26, 4058 Basel,
Switzerland. \email{moritz.gerstung@bsse.ethz.ch}}
\journalname{Dynamic Games and Applications}

\date{\today}

\begin{document}
\maketitle

\begin{abstract}
We analyze the dynamics of evolutionary games in which fitness is defined as an affine function of the expected payoff and a constant contribution. The
resulting inhomogeneous replicator equation has an homogeneous equivalent with
modified payoffs. The affine terms also influence the stochastic dynamics of a two-strategy Moran model of a finite population. We then apply the affine fitness function in a model for
tumor-normal cell interactions to determine which are the most successful tumor
strategies. In order to analyze the dynamics of concurrent strategies within a tumor population, we extend the model to a three-strategy game involving distinct tumor cell types as well as normal cells. In this model, interaction with
normal cells, in combination with an increased constant fitness, is the most
effective way of establishing a population of tumor cells in normal tissue.


\end{abstract}

\newpage
\section{Introduction}
\label{sec:intro}
Evolutionary dynamics describes changes in populations of
competing individuals over time \cite{Nowak2006}. These changes depend on the notion of fitness, a quantity that describes how many offspring a member of the population is
expected to produce. In the simplest model of fitness, the number of offspring depends only on the individual itself and not on 
other individuals or the environment. More generally, fitness may be modeled as frequency-dependent, accounting for interactions among individuals. In 
evolutionary game theory, fitness is modeled as the outcome of a game whose players adopt distinct strategies; in this framework, individuals are identified with the strategy they play
\cite{Maynard-Smith1982,HofbauerBAMS2003}. The fitness assigned to a given
strategy is typically defined as the expected payoff resulting from playing 
the game with all other strategies present in the population, and in this case,
fitness is a linear function of the frequencies. Recently, nonlinear fitness 
functions have also been discussed \cite{TaylorTPB2006,Prugel-Bennett:1994fk,altrock2009deterministic,Traulsen:2008pt,Traulsen:2007kh}.

An example of an evolving system is the cell population of a tumor. Tumors arise from normal cells
in an organism through mutations that increase their somatic fitness, which leads to
outgrowth of normal tissue by the tumor and eventually to invasion of other organs \cite{CairnsN1975,
NowellS1976}. The increased proliferation of cancer cells is, in part, due to
interactions with normal cells \cite{AxelrodPNAS2006}. One example of
tumor-stroma interactions is vascular endothelial growth factor (VEGF) signaling \cite{MuellerNRC2004}.
Many tumors secrete the mobile growth factor
VEGF which stimulates the production of blood vessels. Angiogenesis, in turn, 
increases the fitness of tumor cells through the supply with nutrients and
oxygen.  Using experimental and bioinformatics methods, it has recently been 
estimated that cancer cells make up only 49\% of the cells in tumor  tissue
\cite{Van-LooPNASUSA2010}. The surprisingly high fraction of normal cells  in a
tumor indicates that normal cells play an important role in tumor development. 
However, it remains elusive to which extent the interaction between normal and 
tumor cells may contribute to the proliferative advantage of tumor cells.

To quantify the somatic evolution of tumors mathematical models are
used \cite{MichorNRC2004}. Approaches include population genetics models
\cite{BeerenwinkelPCB2007,DurrettAAP2009,GerstungMPS2010,BozicPNASUSA2010} that describe the
accumulation of driver mutations which confer a fitness advantage to the tumor cells, and
evolutionary game theory models \cite{BasantaSTICM2008}. Game-theoretic
approaches were used to describe both interactions of tumor cells with the environment \cite{GatenbyCR2003} as
well as among tumor cells \cite{TomlinsonEJC1997}. Interactions are not
restricted to be pairwise. For example, \citet{DingliBJC2009} recently analyzed
the joint interactions of multiple myeloma cells with osteoclast and osteoblast
cells in the framework of evolutionary games. 

In the present work, we model fitness to be composed both of a game-theoretic
interaction term and a constant term specific to each cell type. This choice is
motivated by the fact that cancer cells harbor multiple mutations, which can affect cell-cell interactions, alter intrinsic behavior, or both. Hence, the fitness function is an affine function of the
relative frequencies of normal and tumor cells.

We first analyze the evolutionary dynamics in general in the framework of the 
replicator equation. Specifically, we clarify the role of the interaction 
term relative to the constant fitness term. We show that the Prisoner's Dilemma
game, which does not allow for the evolution of cooperation, is transformed by 
adding a constant fitness term in such a way that cooperation becomes possible for certain parameter choices of the affine
fitness function. We also analyze how the affine fitness terms affect the stability criteria for a Moran model of a two-strategy game in a finite population.

The results for the continuous replicator model are then applied to
assess whether exploitation of normal cells or intrinsic proliferation is more 
evolutionarily favorable for a tumor cell. We analyze a set of tumor
strategies leading to the same equilibrium with normal cells in a pairwise game,
and find that the strategy with both a constant fitness advantage and attraction
of normal cells is most successful in the competition with other tumor
strategies.

\section{Inhomogeneous evolutionary games in infinite populations}
\label{sec:general}

We consider a game with $n$ strategies and payoff matrix $\boldsymbol{M} \in
\mathbb{R}^{n\times n}$. The entry $m_{ij}$ of $\boldsymbol{M}$ denotes the payoff to
strategy $i$ if playing against strategy $j$. In evolutionary game theory, a
fixed strategy is associated to each individual. In our application, we think of
the strategy as being determined by the genetic changes of the cancer cell and
we identify strategies with genotypes and with cell types. We first assume an infinite
population size and describe the state of the population by the vector 
\begin{equation}
\mathbf{x} \in
S_{n-1} = \left\{ \mathbf{x} \in [0,1]^n \: \Big| \: \sum_{i=1}^n x_i = 1 \right\}
\end{equation}
of
relative strategy frequencies. The state space $S_{n-1}$ is the
$(n-1)$-dimensional probability simplex. The fitness of a type $i$ individual is
the expected payoff 
\begin{equation}
f_i(\mathbf{x}) = \sum_{j=1}^n m_{ij} x_j.
\end{equation}
Let us assume now that fitness is composed of such a linear term arising from a
game plus a constant term $r_i \in \mathbb{R}$. In vector notation, the
resulting affine fitness function is
\begin{equation}
\label{eq:fitness}
f(\bx) = \boldsymbol{M} \mathbf{x} + \mathbf{r}
\end{equation}
where $\mathbf{r} \in \mathbb{R}^n$.
For $\mathbf{r}=(0,\ldots,0)^\top$, we recover the strong selection limit,
where fitness is directly given by the expected payoff of the game.
For $w\in\mathbb{R}_+$, $w \ll 1$, and $\mathbf{r}=(1-w, \dots, 1-w)^\top$,
the affine fitness function can be interpreted as that of the game with payoff $\boldsymbol{M}' := (1/w) \boldsymbol{M}$ in the weak selection limit: $f(\bx) = w\boldsymbol{M}' \bx + (1-w, \dots, 1-w)^\top.$
In both limiting cases, all components of the constant term $\mathbf{r}$ are identical. In the following, we 
relax this constraint and allow the components $r_i$ to be different for each strategy.

For infinite population size, the dynamics of reproducing individuals can be
described by the replicator equation \cite{taylor1978evolutionary,zeeman1980population,SchusterJOTB1983} as
\begin{equation}
\label{eq:replicator}
\dot{x}_i = x_i \left[ f_i(\mathbf{x}) - \phi(\mathbf{x}) \right],
\quad i=1,\dots,n
\end{equation}
where $\phi(\bx) = \bx^\top f(\bx)$ is the average fitness of the population.
The fixed
points of this system are the solutions of
the set of algebraic equations
$\dot \bx = 0$. The replicator equation always has the $n$ trivial solutions $\bx^*$ given by $x^*_i=1$ and $x^*_j =
0$, for all $j\neq i$. \citet{hofbauer1998evolutionary} provide a general proof for the possible number of internal equilibria in two-player, n-strategy games. A more general system is considered in \citet{Gokhale:2010tw}.

For the affine fitness function defined in Eq.~\ref{eq:fitness}, the replicator
equation \ref{eq:replicator} is said to be inhomogeneous. The inhomogeneous
replicator equation can be interpreted as the (homogeneous) replicator equation
of a transformed game.

\begin{theorem}[\citealp*{StadlerB1991}]
The inhomogeneous replicator equation with affine fitness function $f(\bx) = \boldsymbol{M} \bx + \br$ is equivalent
to the homogeneous replicator equation with linear fitness function $f'(\bx) = \boldsymbol{M}' \bx$  with
\begin{equation}
\label{eq:transformed-payoff}
m_{ij}' = m_{ij} + r_i.
\end{equation}
\end{theorem}
\begin{proof}
Because $\sum_{j=1}^n x_j=1$, one has
\begin{align}
f_i'(\bx) = \sum_{j=1}^n(m_{ij} + r_i)x_j & = \sum_{j=1}^n m_{ij}x_j + r_i = f_i(\bx).
\end{align}
It follows that $\phi'(\bx) = \bx^\top f'(\bx) = \phi(\bx)$. Hence we have $f_i(\bx)-\phi(\bx) = f_i'(\bx)-\phi'(\bx)$, which completes the proof.
\end{proof}

Theorem~1 shows that the replicator dynamics induced by an affine fitness
function can be obtained from an equivalent homogeneous replicator equation.
However the evolutionary dynamics of the transformed game $\boldsymbol{M}'$ can be
substantially different from the one based on $\boldsymbol{M}$ alone.

\section{Two-player games in infinite populations}
\label{sec:two-player}

Passing from a linear to an affine fitness function by adding a
constant fitness term shifts the equilibrium of the replicator equation. We
consider the inhomogeneous replicator equation for two types of individuals
(strategies) $A$ and $B$ with
\begin{equation}
\label{eq:2-strategy-Mr}
\boldsymbol{M} =
\bordermatrix{
&A&B\cr
A&a&b\cr
B&c&d
} \qquad \text{and} \qquad
\br = 
\begin{pmatrix}
s \cr t
\end{pmatrix}.
\end{equation}

By Theorem~1, there exists a non-trivial fixed
point in which the proportion of $A$ individuals in the population is given by
\begin{equation}
\label{eq:fixedpoint-n2}
x^* = \frac{\beta - \sigma}{\beta - \alpha},
\end{equation}
where we have defined $\alpha = a - c$, $\beta = b - d$, and $\sigma = t - s$.
The fixed point $x^*$ is attractive and in $(0,1)$ if and only if
\begin{equation}
\label{eq:stability}
\alpha < \sigma < \beta,
\end{equation}
that is, if the difference in constant fitness $\sigma$ is between the
differences in payoff $\alpha$ and $\beta$.
It follows that for any game with $\alpha < \beta$ there exist constant fitness contributions $r$ such that $\sigma$
satisfies~(\ref{eq:stability})
and $x^*$ is a stable, non-trivial equilibrium point.

Apart from the stable equilibrium, dominance of $A$ (denoted by
$B\rightarrow A$), dominance of $B$ ($A \rightarrow B$), and an unstable
equilibrium at $x^*$ are possible. The parameter regimes leading to these
dynamics are summarized in Table~\ref{tab:2-strategies}.

\begin{table}
\centering
\begin{tabular}{ccc}
\hline
Behavior & Schematic & Parameter Range\cr
\hline
 Stable internal equilibrium & $A \rightarrow x^* \leftarrow B$ & $\alpha < \sigma < \beta$ \cr
Unstable internal equilibrium & $A \leftarrow x^* \rightarrow B$ & $\alpha > \sigma > \beta$ \cr
A dominates B & $A \longleftarrow B$ & $\sigma<\alpha,\beta$ \cr
B dominates A & $A \longrightarrow B$ & $\sigma>\alpha,\beta$ \cr
\hline
\end{tabular}
\caption{Stability of the 2-strategy replicator equation with affine fitness
function with game $\boldsymbol{M}$ and offset $\br$, Eq.~\ref{eq:2-strategy-Mr}. Parameters
are as in Eq.~\ref{eq:fixedpoint-n2}. 
Arrows indicate the change in the composition of the population over time. 
}
\label{tab:2-strategies}
\end{table}

\subsection{The Prisoner's Dilemma}
\label{sec:prisoners-dilemma}

The Prisoner's Dilemma is a metaphor for the evolution of cooperation \cite{AxelrodS1981}. This game
is defined by the inequalities
\begin{equation}
b < d < a < c 
\end{equation}
Strategy $A$ is called cooperation and denoted $C$, whereas strategy $B$  is
called defection and denoted $D$. The inequalities imply that $\alpha,\beta <  0$.
For linear fitness functions, the non-trivial fixed point $x^*$ lies outside of
the  unit interval; thus  $x_A = 0$, $x_B=1$ is the only attractive fixed point and 
cooperation cannot evolve in this model.

For affine fitness functions, however, there does exist a stable equilibrium
between cooperators and defectors provided that $\alpha< \sigma <\beta$, as
shown in the previous section (Figure~\ref{fig:PD}). The necessary condition
$\alpha<\beta$ does not hold for all Prisoner's Dilemma games. If
$\beta<\alpha$, then there exists only an unstable fixed point in $(0,1)$
(Figure~\ref{fig:PD}). In both cases, the all-cooperator equilibrium is stable,
since $\sigma < \beta$. It is unique if in addition $\sigma < \min \{\alpha,\beta\}$.
In this regime, constant selection dominates the dynamics of the evolutionary game,
always favouring cooperators over defectors. Conversely, if $\sigma > \alpha$,
there exists a stable all-defectors equilibrium, which is unique if in addition
$\sigma>\max\{\alpha, \beta\}$.

\begin{figure}
\centering 
\includegraphics[width=.49\textwidth]{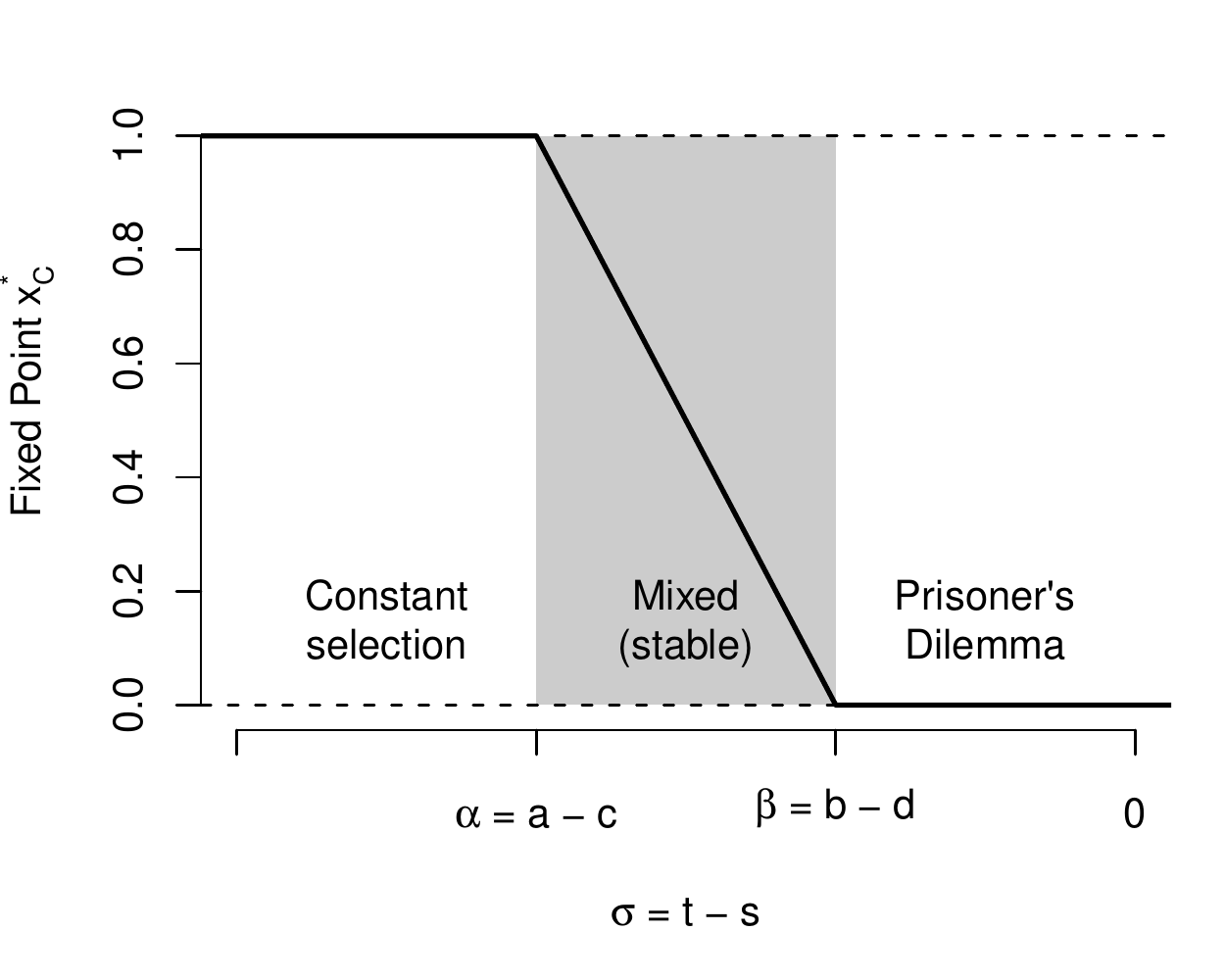}
\includegraphics[width=.49\textwidth]{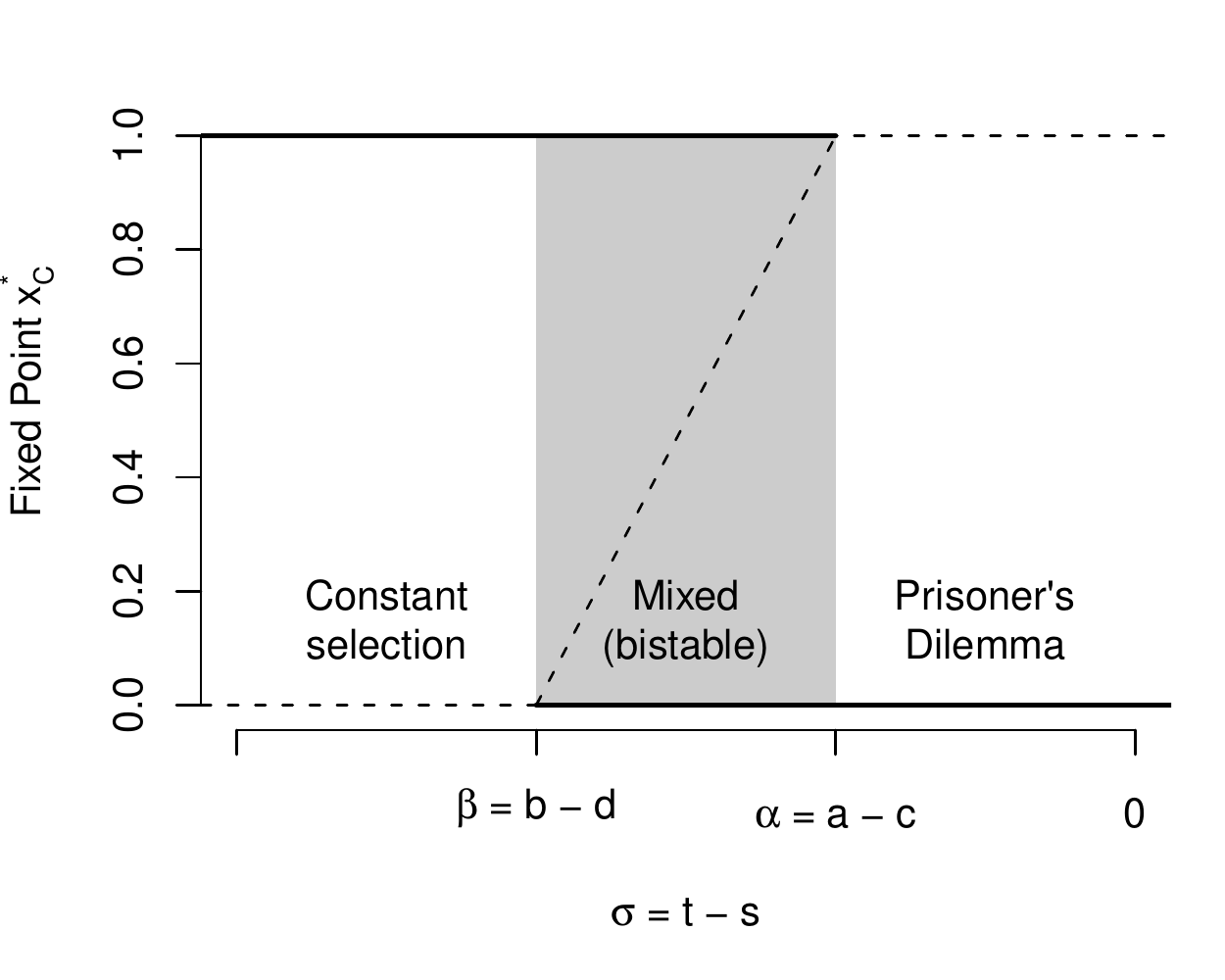}
\caption{Prisoner's Dilemma game with constant fitness advantage $\sigma$ of the
cooperators. The solid black lines denote the frequency of cooperators $x^*_C$
at the stable fixed point $x^*=(x^*_C,x^*_D)^\top$ as a function of
$\sigma=t-s$. Dashed lines are unstable fixed points. For $\alpha < \beta$, there
exists a stable equilibrium (left, grey area) if $\alpha <\sigma<\beta$. For $\beta <\alpha$, the
internal fixed point at $\beta < \sigma < \alpha$ is unstable (right plot). For the homogeneous Prisoner's Dilemma, $\sigma = 0$,
defection is always stable.}
\label{fig:PD}
\end{figure} 

The Prisoner's Dilemma has been studied in many variations in order to derive
conditions under which cooperation can evolve. For example, \citet{Nowak2006}
listed five rules for the evolution of cooperation. To formulate and to
compare these rules, often a simplified version of the Prisoner's Dilemma is
considered, where a cost $\mathfrak{c}$ is paid for cooperation, of which other
cells receive a benefit $\mathfrak{b}$. Defectors do not pay a cost, and other
cells do not receive a benefit from defectors:
\begin{equation}  \label{eq:PD-simple}
\boldsymbol{M} =
\bordermatrix{
&C&D\cr
C&\mathfrak{b}-\mathfrak{c} & - \mathfrak{c} \cr
D&\mathfrak{b} &0
}
\end{equation}
We consider the affine fitness function obtained from this simplified Prisoner's
Dilemma game plus a constant fitness term. Because $\alpha = \beta =
-\mathfrak{c}$,  the replicator dynamics do not allow for stable coexistence
between cooperators and defectors. However, both pure strategies can be
stable, namely cooperation for $\sigma < -\mathfrak{c}$, or equivalently,
\begin{equation}  \label{eq:rule6}
t - s > \mathfrak{c}
\end{equation}
and defection otherwise (Figure~\ref{fig:PD}).
Hence, selection favors cooperators if the constant fitness
advantage is higher than the cost to pay for cooperation.
This rule for the evolution of cooperation makes precise the tradeoff between
constant fitness contributions and those resulting from playing the Prisoner's Dilemma game.

\section{Inhomogeneous two-player games in finite populations}
\label{sec:finite}
The deterministic replicator equation describes the dynamics of infinite populations. In finite populations a stochastic description in which interactions between members of the population are considered individually is more appropriate.

We consider a finite population of constant size $N$, containing as before two types of individuals (strategies) $A$ and $B.$ $\boldsymbol{M}$ and $\br$ are defined as in Eq.~\ref{eq:2-strategy-Mr}. Let $i$ denote the number of $A$ individuals in the population, and denote by $F_i$ and $G_i$ the sum of the expected game payoff and the constant fitness term for types $A$ and $B$ respectively. We have
\begin{align}
F_i &= \frac{(i-1)a + (N-i)b}{N-1} + s = \frac{(i-1)(a+s) + (N-i)(b+s)}{N-1} \\
G_i &= \frac{ic + (N-i-1)d}{N-1} + t = \frac{i(c+t) + (N-i-1)(d+t)}{N-1}.
\end{align}
Also for a finite population the entries of $\boldsymbol{M}$ may be transformed according to Eq.~\ref{eq:transformed-payoff} to an equivalent homogeneous game $\boldsymbol{M}'$. Yet the resulting changes of the stochastic dynamics in a finite population differ from those that arise in an infinite population. In the following, we discuss three different measures for the evolutionary success of a strategy in a finite population and study how these are affected by the constant fitness terms. 

Following \citet{NowakN2004}, we say that $A$ resists invasion by $B$, or that $A$ is an evolutionarily stable strategy (ESS$_N$), if $G_{N-1} < F_{N-1}$, i.e., if one $B$ individual in an otherwise all-$A$ population has lower fitness than the $B$ individuals. This condition may equivalently be written as
\begin{equation}
(N-1)(\alpha - \sigma) > a-b,
\end{equation}
where $\sigma = t - s$, as above. Note that for $N\rightarrow \infty$, this condition is equivalent to $\sigma < \alpha$, from which it follows that $x_A = 0$ is an unstable fixed point of the replicator equation (Table~\ref{tab:2-strategies}). Similarly, one obtains that $B$ is ESS$_N$ if $F_{N-1}>G_{N-1}$, or 
\begin{equation}
(N-1)(\beta - \sigma) < c-d,
\end{equation}
which is for large $N$ again equivalent to $\sigma < \beta$.

The Moran process \cite{Moran1962} provides one model for evolutionary dynamics in finite populations. In each step of the process, one individual is chosen to reproduce with probability proportional to its fitness; its offspring replaces another individual, chosen at random from the population. Denote by $P_{i,j}$ the probability that given $i$ $A$ individuals, one step of the process yields $j$ $A$ individuals. We have
\begin{align}
P_{i,i+1} &= \frac{if_i}{if_i + (N-i)g_i} \cdot \frac{N-i}{N} \\
P_{i,i-1} &= \frac{(N-i)g_i}{if_i + (N-i)g_i} \cdot \frac{i}{N} \\
P_{i,i} &= 1 - P_{i,i+1} - P_{i,i-1},
\end{align}
and $P_{i,j}=0$ otherwise. To compute the fixation probabilities in the presence of an affine fitness term, we consider the limit of weak selection \cite{TaylorBMB2004,antal2006fixation,lessard2007probability,
wu2010universality}. We define the frequency-dependent fitness for types $A$ and $B$, respectively, by
\begin{align}
f_i &= 1 - w + wF_i \\
g_i &= 1 - w + wG_i,
\end{align}
where the parameter $w \in (0,1)$ defines the intensity of selection. Small values of $w$ indicate near-neutral evolutionary pressure \cite{kimura1985neutral,ohta2002near}

The fixation probability of type $A$, i.e., the probability that one $A$ individual will take over an otherwise all-$B$ population, in the Moran process is $\rho_A = [1 + \sum_{k=1}^{N-1} \prod_{i=1}^{k} (g_i/f_i)]^{-1}$ \cite{karlin1975first,TaylorBMB2004}.
For small $w$ a power series expansion yields
\begin{align}
\rho_A 
&= \frac{1}{N} \cdot \frac{1}{1 - w\left( ({xN - y})/{6} - {(N-1) \sigma}/{2} \right)},
\end{align}
where $x = a + 2b - c - 2d = \alpha + 2\beta$ and $y = 2a + b + c - 4d$ \cite{Nowak2006}. Further details about the convergence of the weak selection limit can be found in \citet{wu2010universality}.

For an $A$ mutant with no fitness advantage over $B,$ we have $\rho_A = 1/N.$ We say $A$ is advantageous (AD), if selection favors the fixation of $A$, $\rho_A > 1/N$, or equivalently
\begin{equation}
xN - y - 3(N-1) \sigma > 0.
\end{equation}
For large $N$, this becomes
\begin{equation}
\label{eq:AD-affine}
\sigma < \frac{\alpha + 2\beta}{3}.
\end{equation}
Hence this condition is fulfilled if the condition $\sigma < \alpha,\beta$ for the replicator equation holds.
Equilibrium is reached when $F_i = G_i,$ or equivalently
\begin{equation}
i =  \frac{\beta - \sigma}{\beta - \alpha},
\end{equation}
in agreement with the continuous case, Eq.~\ref{eq:fixedpoint-n2}.

Lastly, we say $A$ is risk-dominant (RD) over $B$ if $\rho_A > \rho_B.$ In the case of weak selection,
this is equivalent to
\begin{equation}
\frac{\alpha+\beta}{2}N - 6(a-d) - (N-1)\sigma >0.
\end{equation}
For large $N$, $A$ is RD if, and only if
\begin{equation}
\frac{\alpha + \beta}{2} - \sigma > 0
\end{equation}
holds. The criteria for AD, ESS$_N$, and RD are summarized in Table~\ref{tab:finite-pop}. Note that all are fulfilled if $\sigma < \alpha$, as required for $A$
 to be evolutionarily stable in the replicator equation.
 
\begin{table}
	\centering
	\begin{tabular}{lcc}
		\hline
		Criterion & Small population & Large population \\[1ex]
		\hline
		ESS$_N$ & $(N-1)(\alpha - \sigma) < a-b$ & $\sigma < \alpha$\\[1ex]
		AD & $xN - y - 3(N-1) \sigma > 0$ & $\sigma < ({\alpha+2\beta}){/3}$ \\[1ex] 
		RD & $ ({\alpha+\beta})N/2 - 6(a-d) - (N-1)\sigma >0$ & $\sigma < ({\alpha+\beta}){/2}$ \\[1ex]
		\hline
	\end{tabular}
	\caption{Criteria for the evolutionary stability of strategy $A$ in finite populations of small and large size.}
	\label{tab:finite-pop}
\end{table}

As $N$ grows large it is also possible to compute the expected time until fixation. As shown by \citet{antal2006fixation}, the fixation time for an AD strategy starting from a single mutant scales like $N \ln N$, whereas that of a neutral strategy scales like $N^2$. The affine modification affects whether a strategy is advantageous but not the scaling of the fixation time.

\subsection{Affine Prisoner's Dilemma in finite populations}
We return to the affine modification of the Prisoner's Dilemma defined in Eq.~\ref{eq:PD-simple}, considering it now in a finite population of size $N$. We have
\begin{equation}
\rho_C = \frac{1}{N} \cdot \frac{1}{1+\frac{w}{2}((N-1)(\sigma + \frak{c}) - \frak{b})}.
\end{equation}
Cooperation is AD, $\rho_C > 1/N$, if 
\begin{equation}
\label{eq:ESS-PD-simple-finite}
-\sigma > \mathfrak{c} - \frac{\mathfrak{b}}{N-1}.
\end{equation}
That is, the constant fitness advantage $-\sigma$ must be larger than the cost of cooperation $\mathfrak{c}$ minus the benefit $\mathfrak{b}$ divided by the population size minus one. It thus appears that  the fixation probability of cooperators is higher in small populations, than in large ones. For $N\rightarrow\infty$ condition Eq.~\ref{eq:ESS-PD-simple-finite} results in $-\sigma > \frak{c},$ again in agreement with the continuous case.

Because we have $\alpha = -\frak{c} = \beta$, the conditions for ESS$_N$ and RD are equivalent to $-\sigma > \mathfrak{c}$. For $\sigma = 0$, cooperation satisfies none of the criteria for evolutionary success. However, in the affine case, each such criterion is fulfilled if the constant fitness advantage  $-\sigma$ is higher than the cost $\mathfrak{c}$. Thus cooperation can also arise in small populations if cooperators can compensate the cost of cooperation by a constant fitness surplus.

\section{Coevolution of tumor and normal cells}
\label{sec:coev}


Tumors present an example of the evolution of defection, because cancer cells
have lost their normal cooperative behavior and defect the host
\cite{MichorNRC2004}. However, experiments indicate that tumors consist of about
50\% of non-cancerous cells. This fraction appears to be consistent among
distinct cancer subtypes \cite{Van-LooPNASUSA2010}, which raises the possibility
that the normal cells are not the result of a random admixture, but constitute
an attractive equilibrium resulting from interactions of tumor and normal cells.

As shown in section~\ref{sec:two-player}, the two-player dynamics only depends
on the difference $\alpha$ and $\beta$, as well as $\sigma$. For a given
attractive equilibrium, there exist thus many parameter combinations. In the
following, we set $x^*= (1/2, 1/2)$, which is reasonably close to the
experimentally observed ratio of 49\% tumor cells \cite{Van-LooPNASUSA2010}. We
observe that this choice implies $\sigma = \frac{1}{2}(\alpha + \beta)$, which reduces
the number of independent parameters to two.

To explore the range of possible parameter combinations, we define the following
three tumor cell strategies to be played against a normal cell type:
\begin{equation}
\begin{array}{llll}
\label{eq:T1-T3}
 T_1:\qquad & \alpha_1=-1\quad& \beta_1=1\quad & \sigma_1=0 \cr
 T_2: & \alpha_2= -2& \beta_2=0 &  \sigma_2=-1 \cr
 T_3: & \alpha_3= 0& \beta_3=2 & \sigma_3=1
\end{array}
\end{equation}
When played pairwise against normal cells all three strategies 
lead to a stable equilibrium at $x^*_N=1/2$. Tumor strategy $T_1$ presents a
mixed strategy of exploitation and attraction of normal cells, without an additional intrinsic
fitness advantage. $T_2$ is a strategy that strongly exploits (Prisoner's
Dilemma),  however at a the cost of a disadvantage in terms of the constant 
fitness contribution. $T_3$ is a mixed strategy that has both a constant fitness
advantage  and the ability to recruit healthy cells.

\subsection{Three-player games}
\label{seq:3-player}

In large tumors a huge genetic diversity is
generated due to the large number of cells and a potentially increased mutation
rate \cite{BeerenwinkelPCB2007,BozicPNASUSA2010}. It is thus likely that many
different tumor cell types with specific strategies are present in the tumor
simultaneously. Hence a strategy able to dominate many others will be
successful in taking over a tumor.

Let $H$ denote the normal (healthy) cell type and $T_1$ a tumor strategy. Again there is
an affine payoff function for the tumor-normal interaction with the payoff
matrix $\boldsymbol{M}$ and constant fitness $\br$. Now consider a second tumor strategy $T_2$. Assuming
no interactions among the tumor cell types, the payoff matrix and fitness
vector for the three strategies are

\begin{equation}
\boldsymbol{M} = 
\bordermatrix{ 
& H & T_1 & T_2 \cr 
H & a & b_1 & b_2 \cr
T_1 & c_1 & d_1 & 0 \cr
T_2 & c_2 & 0 & d_2
},
\qquad
\br=
\begin{pmatrix}
s \\ t_1 \\ t_2
\end{pmatrix}.
\end{equation} 
According to Theorem~1 the affine fitness function can be rewritten in terms of
the game
\begin{equation}
\label{eq:3-player-M-prime}
\boldsymbol{M}' = 
\bordermatrix{ 
& H & T_1 & T_2 \cr 
H & 0 & 0 & 0 \cr
T_1 & -\alpha_1 + \sigma_1 & -\beta_1 + \sigma_1 & -b_2 + \sigma_1 \cr
T_2 & -\alpha_2 + \sigma_2 & -b_1 + \sigma_2 & -\beta_2 + \sigma_2 \cr
},
\end{equation}
where the first row has been subtracted from all rows to obtain a representation
in terms of $\alpha_i$, $\beta_i$, and $\sigma_i$, $i=1,2$. Interestingly, the constant
fitness terms lead to an effective interaction term among the tumor strategies.
Moreover, the interactions, $m'_{T_1,T_2} = -b_1 + \sigma_2$ and $m'_{T_2,T_1} =
-b_2 + \sigma_1$, depend on the absolute value of $b_1$ and $b_2$, i.e., on the
payoff that a tumor pays to the fitness of normal cells.

We now investigate the replicator dynamics for the transformed game defined
in Eq.~\ref{eq:3-player-M-prime} and the tumor strategies defined in
Eq.~\ref{eq:T1-T3}, following \citet{hofbauer1998evolutionary,bomze1983lotka}. One has $\alpha_i-\sigma_i = -1$, and $\beta_i -
\sigma_i = 1$, $i=1,\ldots,3$. We also
assume that $b_1=b_2=:b$, any difference can be subsumed into the parameters $\beta_i$. It follows that
\begin{equation}
\label{eq:3-player-T-M-prime}
\boldsymbol{M}' = 
\bordermatrix{ 
& H & T_1 & T_2 \cr 
H & 0 & 0 & 0 \cr
T_1 & 1 & -1 & -b - \sigma_1 \cr
T_2 & 1 & -b - \sigma_2 & -1 \cr
}
\end{equation}

This game has the non-trivial fixed
points $x^*=(1/2,1/2,0)^\top$, and $x^*=(1/2,0,1/2)^\top$. However, it depends
on the parameter $b$ whether these are stable or saddle points.
As shown in Appendix~A, there exists an additional fixed point in the interior
of $S_2$ at
\begin{equation}
\label{eq:internal-fixedpoint}
x^* = (\omega_1,\omega_2,\omega_3)^\top/\omega 
\end{equation}
with
\begin{align}
\label{eq:omega_i}
\omega_1 &= 1 - (b - \sigma_1)(b - \sigma_2) \cr
\omega_2 &= 1 - (b - \sigma_1)
\cr \omega_3 &= 1 - (b - \sigma_2),
\end{align}
and $\omega = \sum\omega_i$, if, and only if, all $\omega_i$
have the same sign, $\sgn \omega_1 = \sgn \omega_2 = \sgn \omega_3$
\cite{StadlerBMB1990}.

The dynamic behavior of all three player games $HT_iT_j$ of the normal cell
type with two of the tumor strategies $T_1$, $T_2$, and $T_3$ can be divided
into the following four cases:
\renewcommand{\theenumi}{\roman{enumi}}
\renewcommand{\labelenumi}{(\theenumi)}
\begin{enumerate}
   \item For
	\begin{equation}
	b < \frac{\sigma_1+\sigma_2}{2} -
	\sqrt{\left(\frac{\sigma_1+\sigma_2}{2}\right)^2 - (1-\sigma_1\sigma_2)}
	\end{equation}  
	the normal cell type goes extinct. There exist only a stable equilibrium of the
	two tumor strategies. In this case the constant fitness advantages of both
	tumor types are so large that despite the payoff $b$ to the normal cell type,
	$H$ dies out.
  \item In the regime
  \begin{equation}
  \frac{\sigma_1+\sigma_2}{2} -
\sqrt{\left(\frac{\sigma_1+\sigma_2}{2}\right)^2 - (1-\sigma_1\sigma_2)} < b < 1
+ \min\{\sigma_1,\sigma_2\}
  \end{equation}
  there exists a stable fixed point in the interior of the simplex. Hence all
  three cell types coexist. In this regime the constant fitness advantage of
  both tumor types are smaller than the fitness but not too large to have the
  normal cell type go extinct.

 \item If 
  \begin{equation}
  1 + \min\{\sigma_1,\sigma_2\} < b < 1 +
  \max\{\sigma_1,\sigma_2\}
  \end{equation} there exists no interior fixed point. Then only the fixed point $x_H=1/2$,
  $x_{T_i}=1/2$, for $\sigma_i > \sigma_j$ is stable. That is, the tumor
  strategy that has the larger fitness advantage will win. Equivalently, we can
  write $\min\{t_1-s,t_2-s\}< b-1 <\max\{t_1-s,t_2-s\}$. In this
  case, the constant fitness advantage of one tumor type is smaller than the
  payoff $b$ minus $1$, whereas the other tumor types' fitness advantage exceeds
  this value.

	  \item Lastly, for 
  \begin{equation}
  b> 1 + \max\{\sigma_1,\sigma_2\}
  \end{equation}
  there exists an interior saddle
  point and both fixed points at the edges, $x^*=(1/2,1/2,0)^\top$, and
  $x^*=(1/2,0,1/2)\top$ are stable. Which tumor strategy wins, depends on
  which of the two tumor strategies emerged first. The condition is equivalent to
  $b>1+\max\{t_1,t_2\}-s$. In this regime, the payoff to a normal cell $b$ is
  larger than the constant fitness advantages to both tumor types plus one.
  Therefore, both tumor types attract normal cells more strongly than their
  constant fitness advantage.
	
\end{enumerate}

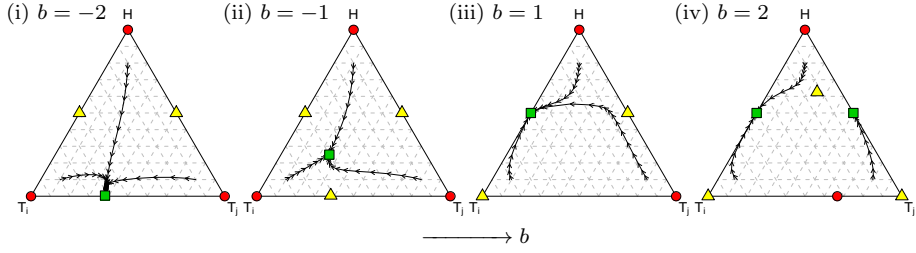
\begin{figure}
	\centering
	\setlength{\unitlength}{\textwidth}
	\begin{picture}(1,0.27)
		\put(0.0,0.02){\includegraphics[width=.25\textwidth, trim=10 10 10 10,
		clip=true]{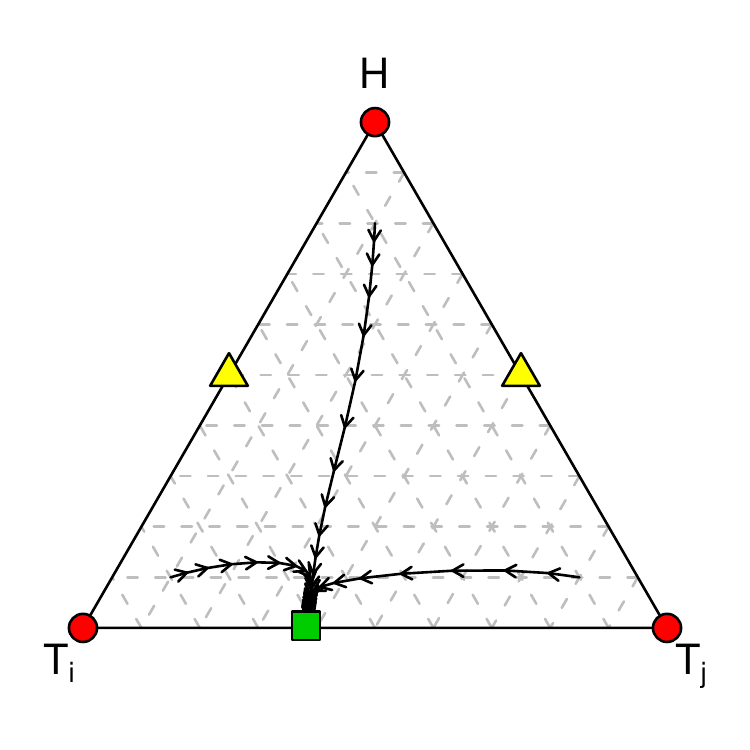}} 
		\put(.25,0.02){\includegraphics[width=.25\textwidth,
		trim=10 10 10 10, clip=true]{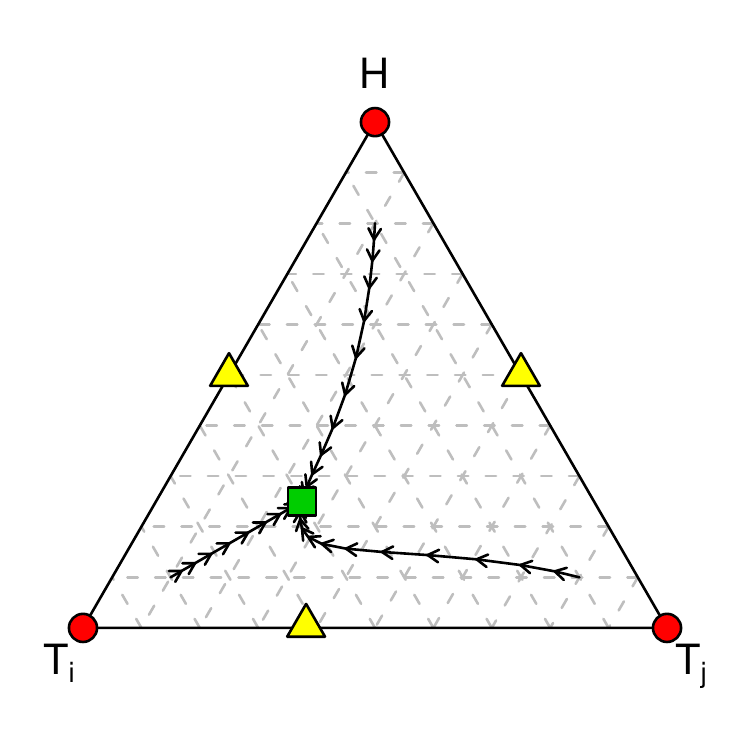}}
		\put(.5,0.02){\includegraphics[width=.25\textwidth, trim=10 10 10 10,
		clip=true]{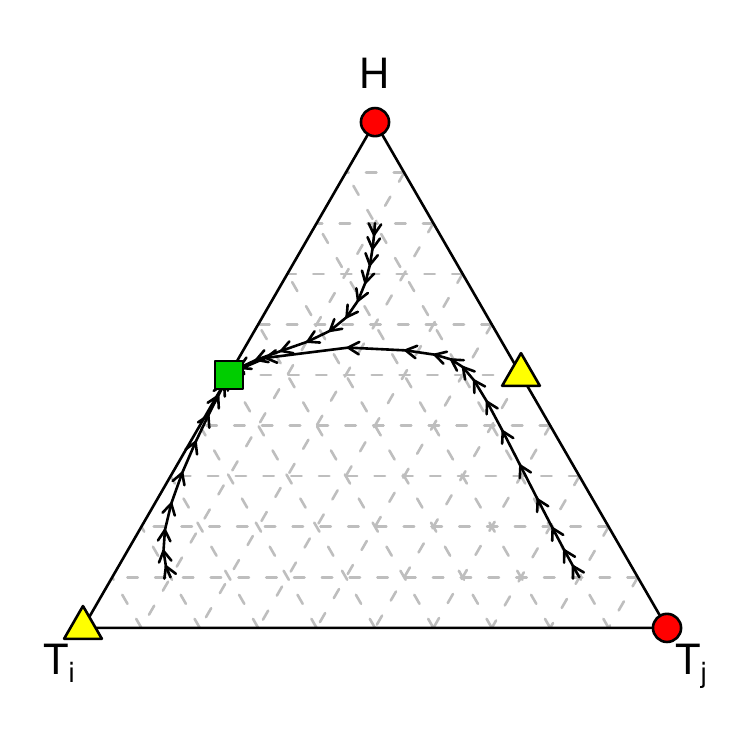}}
		\put(.75,0.02){\includegraphics[width=.25\textwidth, trim=10 10 10 10,
		clip=true]{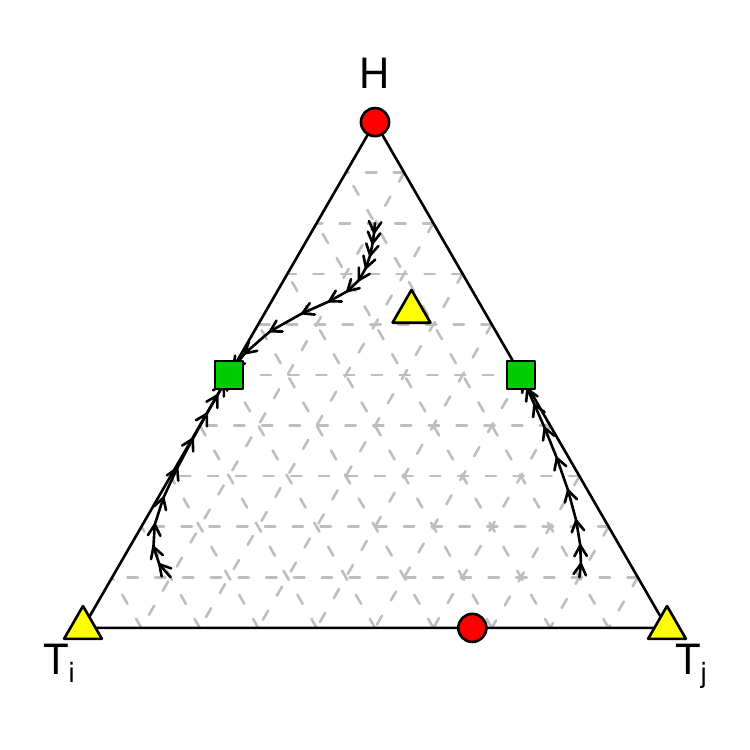}}
		\put(-.01,.25){(i) $b=-2$}
		\put(0.23,.25){(ii) $b=-1$}
		\put(0.48,.25){(iii) $b=1$}
		\put(0.73,.25){(iv) $b=2$}
		\put(0.45,0){$\xrightarrow{\hspace*{1cm}} b$}
	\end{picture}
	\caption{Illustration of the dynamic regimes (i)--(iv). Specific values are for
	the game $NT_1T_2$ for different values of the parameter $b$. Red circles denote unstable fixed points, yellow triangles are saddle points, and green squares are stable equilibria. For case (iii), the fixed point on the edge $NT_i$ for the tumor strategy $T_i$ with $\sigma_i > \sigma_j$ is stable.}
	\label{fig:class-dynam}
\end{figure}

\begin{figure}
\setlength{\unitlength}{\textwidth}
\begin{picture}(1,1.333)
\put(0,1){\includegraphics[width=\textwidth]{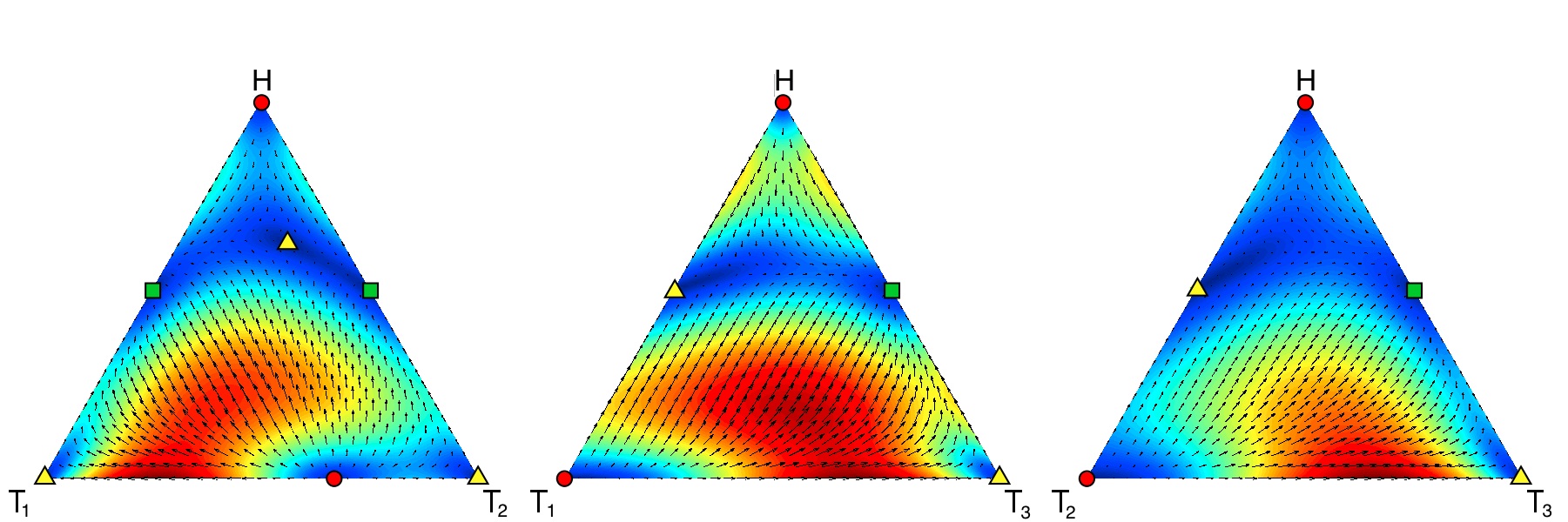}}
\put(0,.66){\includegraphics[width=\textwidth]{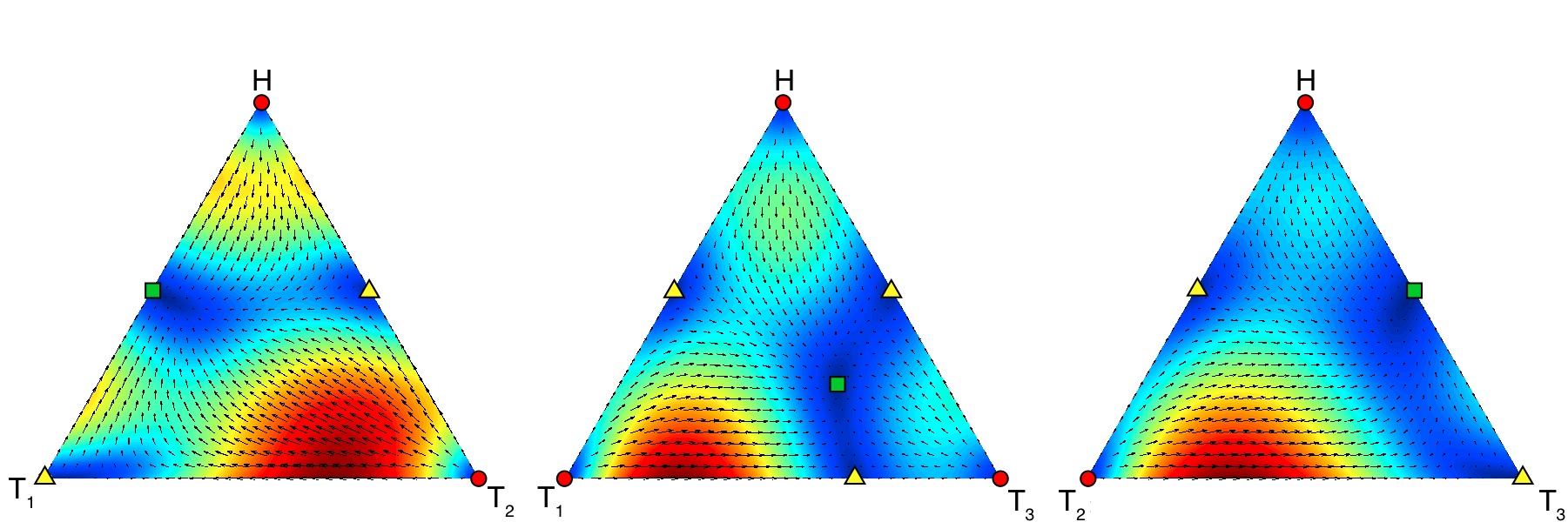}}
\put(0,.33){\includegraphics[width=\textwidth]{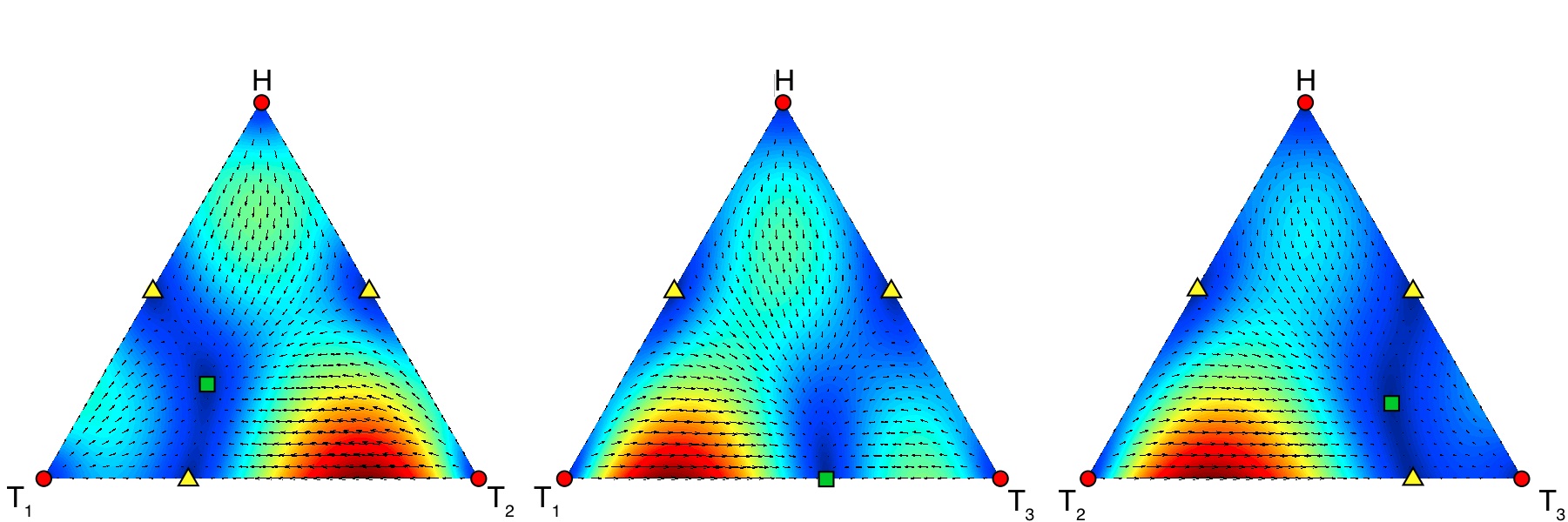}}
\put(0,0){\includegraphics[width=.33\textwidth, trim=10 10 10 10,
clip=true]{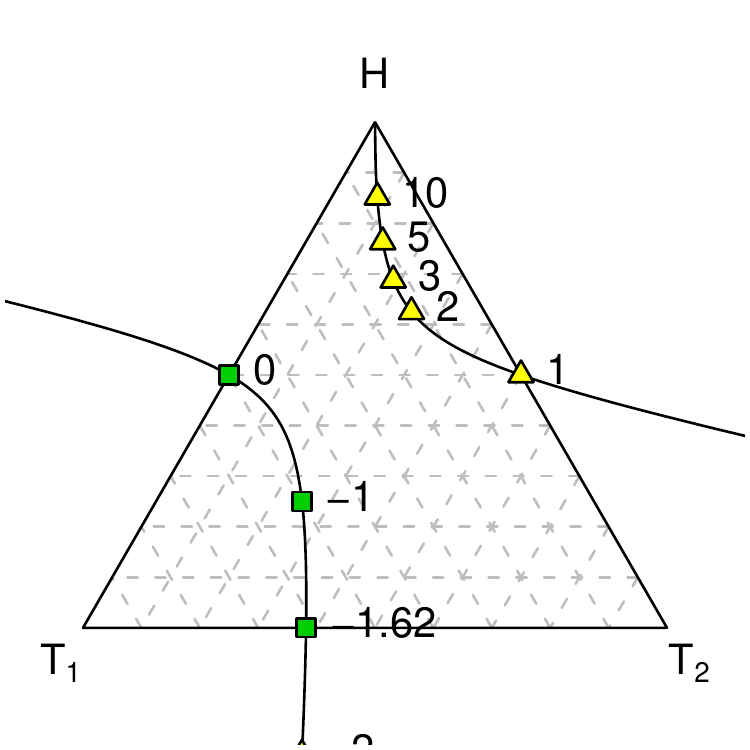}}
\put(.333,0){\includegraphics[width=.33\textwidth,
trim=10 10 10 10, clip=true]{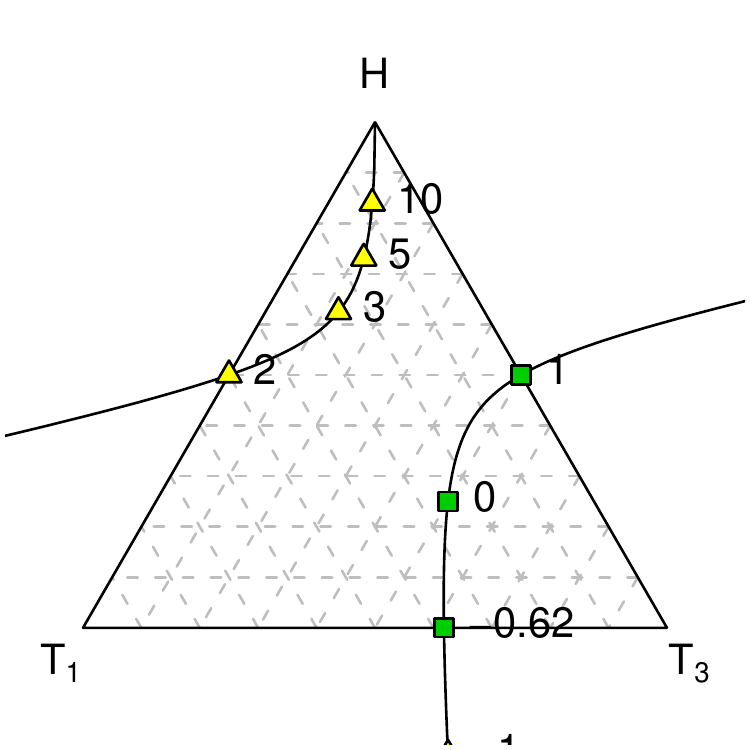}}
\put(.666,0){\includegraphics[width=.33\textwidth, trim=10 10 10 10,
clip=true]{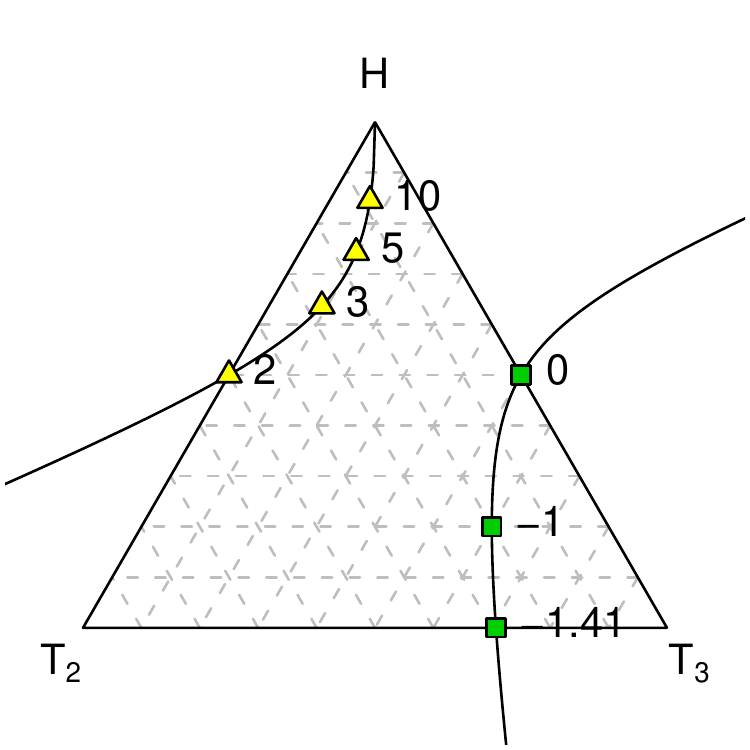}}
\renewcommand{\sffamily}{\fontfamily{phv}}
\put(0,1.26){\large A. $b=2$}
\put(0,0.93){\large B. $b=0$}
\put(0,0.60){\large C. $b=-1$}
\put(0,.28){\large D.}
\end{picture}

\caption{A-C: Tumor strategies in the 3 player case for $b=2,0,-1$.
Arrows denote the direction, and colors the rate of change defined by the
replicator equation, Eq.~\ref{eq:replicator}. Points label the fixed points of
the system, obtained by solving the algebraic equation defined by $\dot x = 0$. Red circles are repelling fixed points (two real
parts of the eigenvalues $\lambda_i$, $i=1,
\dots 3$, of the Jacobian $J(x^*)$ positive),
yellow triangles: saddle points (1 real part positive), and green squares are stable fixed points
with two real parts negative. D: Position and stability of the interior fixed
point as a function of $b$ for all three games. Again yellow dots denote saddle
points, green circles are stable fixed points.}
\label{fig:3player}
\end{figure}
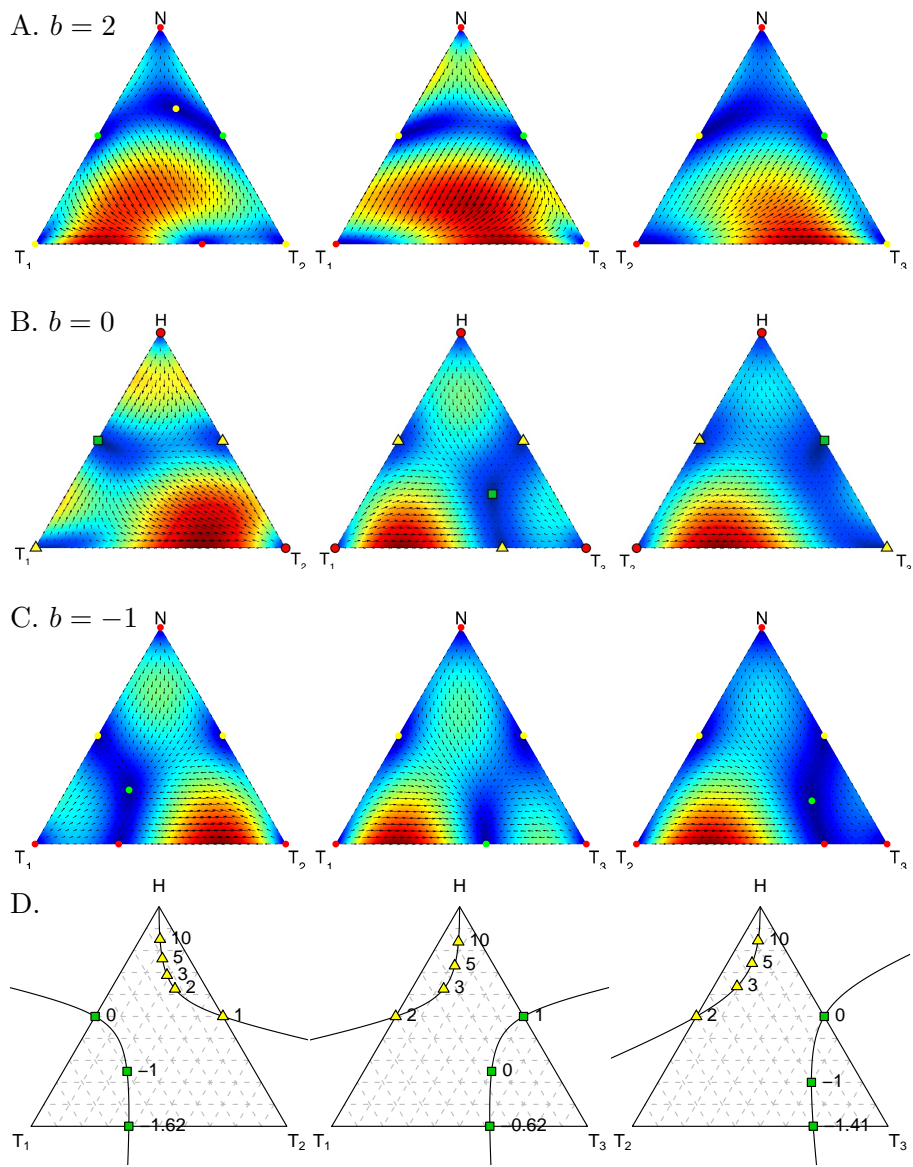

The dynamics associated with each case are illustrated in
Figure~\ref{fig:class-dynam}. These different cases are realized for each game
at different parameter values of $b$. For $b=2$ (Figure~\ref{fig:3player}A), we have case (vi) for the game $HT_1T_2$, and case (iii) for games $HT_1T_3$, $HT_2T_3$. Thus for $HT_1T_2$, whichever of the strategies $T_1$ and $T_2$ arises first
will be successful and converges to an equilibrium with $H$. In the other
two games, however strategy $T_3$ wins over $T_1$, and $T_2$, respectively. We
thus conclude that $T_3$ is most successful. The coordinates of the equilibrium interior equilibrium $x^*$ defined in
Eq.~\ref{eq:internal-fixedpoint} are depicted in Figure~\ref{fig:3player}D for
the three games $HT_1T_2$, $HT_1T_3$, and $HT_2T_3$ as a function of the
parameter $b$.

Because tumor strategy
$T_3$ has the largest constant fitness advantage, $\sigma_3 > \sigma_1 >
\sigma_2$, it outcompetes the other tumor strategies for most values of $b$. If $b$ is
of moderate size, there exists an equilibrium with $T_1$
(Figure~\ref{fig:3player}B). The equilibrium level lies, however, at low
frequencies of $T_1$. There also exists a stable equilibrium with $T_2$ for
$b<0$ (Figure~\ref{fig:3player}C). Yet the latter appears biologically
unrealistic, as it would effectively repel normal cells, which contradicts the
consensus of a positive interaction \cite{MuellerNRC2004}.

Tumor strategy $T_1$, which both exploits and attracts normal cells, wins over
$T_2$ for $b>0$. This is because of the fitness disadvantage of $T_2$. In comparison with $T_3$, there is a bistable equilibrium
for $b>2$. For $0< b < 2$, however, $T_3$ dominates $T_2$.

The affine Prisoner's Dilemma strategy $T_2$ performs worst of all, because it
has a fitness {disadvantage} compared to normal cells, in contrast to the
other two tumor strategies. This disadvantage is required to generate a stable
equilibrium with normal cells. Due to this disadvantage, however, the strategy is 
easily outcompeted by $T_1$ and $T_3$.

\section{Discussion}
In evolutionary games with affine fitness functions, the corresponding inhomogeneous replicator equation has an
equivalent homogeneous replicator equation with a transformed game. The
transformations of the game, however, can cause substantial differences as
compared to the original game. For $n=2$ dimensions, the resulting changes can be fully characterized in terms
of the constant fitness difference $\sigma$. 
The affine transformations of the payoff matrix also influence the stochastic Moran model of a finite population, and we have evaluated how different criteria for the stability of a strategy are affected.

The affine terms also
alter the dynamics of the Prisoner's Dilemma game, a prototype game for
studying the evolution of cooperation \cite{AxelrodS1981}. In the presence of a constant fitness
advantage, cooperation can arise if the fitness advantage is higher than the
cost to pay for cooperation. This simple rule adds to the five rules for the evolution of cooperation that were
presented recently \cite{NowakS2006}. It makes precise the intuitive idea that cooperators can evolve if they compensate their disadvantage in the game by an intrinsic fitness contribution.

In mathematical and computational models of cancer, reviewed recently by \citet{attolini2009evolutionary}, evolutionary game theory provides a useful tool for analyzing the role of cell differentiation and heterogeneity in tumor initiation and progression. In the earliest studies applying evolutionary game theory to cancer, Tomlinson and Bodmer investigated the effects on tumor progression of several cell behaviors, including cytotoxicity, angiogenesis, and apoptosis \cite{TomlinsonEJC1997,tomlinson1997modelling}. Further studies have used evolutionary game theory to model spatial dynamics \cite{bach2003spatial}, tumor-host interaction \cite{GatenbyCR2003}, and the interaction of multiple tumor cell phenotypes \cite{Bach:2001fu}, in particular those of invasive or motile cells \cite{mansury2006evolutionary,BasantaSTICM2008,basanta2008evolutionary}. In the present study, we take explicit account of the effects of affine transformations on evolutionary games in cancer, and we use this framework to investigate interactions among cell types within the tumor population.

Tumor cells present an example where defecting strategies arise in an
organism of cooperating cells. Experimental data indicates that tumor cells
exist in a stable equilibrium with normal cells \cite{Van-LooPNASUSA2010}. We
have therefore analyzed what type of tumor strategies would lead to such coexisting states by means of an
inhomogeneous replicator equation. Fitness is modeled to contain both a
game-theoretic and a constant term. Because most parameters of the model
cannot be directly obtained by experiment, we assessed three different prototype tumor strategies leading
to the same equilibrium level. To get further insight into the relative 
contributions of both the interactions with normal cells, and the intrinsic tumor-specific
fitness, we have analyzed the dynamics of multiple tumor strategies. This
approach is motivated by the finding that in large populations, new cell types
arise quickly through spontaneous mutations. It is thus a requirement for a
winning strategy to be able to compete with many others, although many of them are likely to exist only in very low frequencies.

The analysis of multiple tumor strategies in a three player game shows that the
affine fitness function introduces correction terms that cause an effective
interaction of the tumor cells. The strength of these interaction terms was
given by the constant fitness advantage, and the absolute payoff of tumor to
normal cells. We have then classified the dynamics of the system based on the
 payoff to normal cells and the constant fitness advantage. We find that the
 most successful tumor strategy has both a constant fitness advantage and a
 payoff to normal cells (relative to the payoff to itself). 
 
 The payoff to normal
 cells could be mediated by a mobile growth factor such as VEGF. This growth
 factors is secreted by tumor cells to attract blood vessels that, in turn,
 supply the tumor with nutrients \cite{MuellerNRC2004}. Interestingly, the interaction through
 VEGF is also a target for drug interventions yet with
 ambivalent success \cite{CarmelietN2005}. Our analysis also elucidates
 when such an intervention would be successful: In the replicator dynamics, an
 equilibrium between tumor and normal cells exists only if $a-c< t-s < b-d$.
 A therapeutic success would occur if normal cells dominate the dynamics. This
 requires $a-c$ and $b-d$ to be larger than $t-s$. An VEGF inhibitor would
 reduce the parameter $b$, ideally to zero. This is however, not sufficient for
 the replicator dynamics to favour normal cells, because also $a-c$ must become
 larger than $t-s$. A successful therapy must additionally reduce $c$ to zero such
 that $t-s<a$. However, this may be difficult, or even impossible,
 to achieve for tumor cells with a high constant fitness advantage.
 
In our model of tumor-normal interactions, all strategies were
given. In cancer, however, new
strategies are thought to arise through mutations. In the future it could thus
be an interesting extension to the model to include mutations
that transform one strategy into another, see for example
\cite{FudenbergTPB2006} for a general analysis of evolutionary game
theory with mutations in finite populations. Such an extension may
also be capable of assessing the interactions with cancer stem cells, which
are hypothesized to form a distinct tumor
subpopulation that replenishes normal tumor cells
\cite{wicha2006cancer,clarke2006cancer}. 

\appendix

\section{Fixed points of three player strategies}
\label{sec:appendix-a}
Rewrite matrix $\boldsymbol{M}'$ of Eq.~\ref{eq:3-player-M-prime} in normal form:
\begin{equation}
\label{eq:3-player-M-prime-normal}
\boldsymbol{M}' = 
\bordermatrix{ 
& N & T_1 & T_2 \cr 
N & 0 & \beta_1 - \sigma_1 & \beta_2 - \sigma_2 \cr
T_1 & -\alpha_1 + \sigma_1 & 0 & \beta_2 -b_2 -\sigma_2 + \sigma_1 \cr
T_2 & -\alpha_2 + \sigma_2 &  \beta_1-b_1 - \sigma_1 + \sigma_2 & 0 \cr }
=: 
\begin{pmatrix}
 0 & \delta & \gamma \cr
 \alpha & 0 & \epsilon \cr
\eta & \beta  & 0 \cr 
\end{pmatrix}.
\end{equation}
Now define:
\begin{align}
\omega_1 &= \delta\epsilon + \gamma\beta - \epsilon\beta\\
\omega_2 &= \alpha\gamma + \epsilon\eta - \gamma\eta\\
\omega_3 &= \eta\delta + \alpha\beta - \alpha\delta
\end{align}
One finds:
\begin{align}
\omega_1 &= (\beta_1 - \sigma_1)(\beta_2 - \sigma_2) - (b_1 - \sigma_2)(b_2 -
\sigma_1) \\
\omega_2 &= (\alpha_2 - \sigma_2)(b_2 - \sigma_1) - (\alpha_1 -
\sigma_1)(\beta_2 - \sigma_2) \\
\omega_3 &= (\alpha_1 - \sigma_1)(b_1 - \sigma_2) - (\alpha_2 -
\sigma_2)(\beta_1 - \sigma_1),
\end{align}

It can be shown \cite{StadlerBMB1990} that there exists a fixed point in the
interior of $S_2$, iff all $\omega_i$, $i=1,\ldots,3$ have the same sign, $\sgn \omega_i = \Sigma$
\cite{StadlerBMB1990}.
Its coordinates are given by $x^* = (\omega_1,\omega_2,\omega_3)^\top/\omega$.
The stability of $x^*$ is determined by the determinant $\Delta =
\alpha\beta\gamma + \delta\epsilon\eta$, that is:
\begin{align}
\Delta &= (\alpha_1 - \sigma_1)(b_1-\sigma_2)(\beta_2-\sigma_2) +
(\alpha_2-\sigma_2)(b_2 - \sigma_1)(\beta_1-\sigma_1)\cr
&\qquad -
(\beta_1-\sigma_1)(\beta_2-\sigma_2)(\alpha_1-\sigma_1+\alpha_2-\sigma_2).
\end{align}
The interior fixed point is stable iff both eigenvalues of the Jacobian,
$\lambda_{1/2} = -\Sigma (\Delta \pm \sqrt{\Delta^2 -
4\omega_1\omega_2\omega_3})$, have negative real parts. This requires that
$\sgn \Delta >0 $, and $\Sigma > 0$. Hence all $\omega_i$ need to be positive.

These conditions simplify for the tumor strategies defined in
Eq.~\ref{eq:T1-T3}. These imply $\alpha_i-\sigma_i = -1$, and $\beta_i -
\sigma_i = 1$, $i=1,\ldots,3$. With the additional assumption $b_i=b$,
$i=1,\ldots, 3$, it follows that
\begin{align*}
\omega_1 &= 1 - (b - \sigma_1)(b-\sigma_2) \cr
\omega_2 &= 1 - (b - \sigma_1) \cr
\omega_3 &= 1 - (b - \sigma_2).
\end{align*}
The determinant reads
\begin{equation}
\Delta = 2 - (b - \sigma_1) -(b-\sigma_2) = \omega_2 + \omega_3.
\end{equation}
Hence the condition $\sgn \Delta = \Sigma$ is fulfilled if $\sgn\omega_1 =
\sgn\omega_2 = \Sigma$. The conditions for each $\omega_i$ to be positive are:
\begin{align}
\omega_1 > 0 \qquad \Leftrightarrow \qquad & \left| b
+\frac{\sigma_1+\sigma_2}{2} \right| < 
\sqrt{\left(\frac{\sigma_1+\sigma_2}{2}\right)^2 - (1-\sigma_1\sigma_2)}
\\
\omega_2 > 0 \qquad \Leftrightarrow \qquad &b > 1+\sigma_1 \\
\omega_3 > 0 \qquad \Leftrightarrow \qquad &b > 1+\sigma_2 
\end{align}
For the parameter values of the tumor strategies $T_i$ considered here,
$\omega_1 < 0$ if both $\omega_2$, and $\omega_3$ are negative. It thus follows
that all three $\omega_i$ are negative, iff $b > 1+ \max\{\sigma_1,\sigma_2\}$.
Then the interior fixed point $x^*$ is unstable. For the interior fixed point
to be stable, all three conditions Eq. 33--35 need to be fulfilled.

\bibliography{lit,hani}
\bibliographystyle{genres}

\end{document}